\documentclass[11pt,letterpaper]{article}
\usepackage[top=1in,left=1in,bottom=1in,right=1in]{geometry}
\usepackage{graphicx} 
\usepackage[colorlinks=True,allcolors=blue]{hyperref}

\usepackage{amssymb,amsmath,amsfonts}
\usepackage{amsthm}
\usepackage{xspace}
\usepackage{algorithm}
\usepackage{algpseudocode}

\DeclareMathOperator*{\argmin}{arg\,min}
\DeclareMathOperator*{\argmax}{arg\,max}
\DeclareMathOperator{\F}{\mathcal{F}}
\DeclareMathOperator{\T}{\mathcal{T}}
\DeclareMathOperator{\Z}{\mathbb{Z}}
\DeclareMathOperator{\maaf}{MAAF}

\newcommand{\set}[1]{\{#1\}}
\newcommand{\OLA}{\texttt{OLA}}
\newcommand{\multi}{\texttt{multi}}
\newcommand{\mset}{\texttt{m\_set}}
\newcommand{\mroot}{\texttt{m\_root}}
\newcommand{\leafof}{\texttt{leaf\_of}}

\newtheorem{theorem}{Theorem}
\newtheorem{lemma}{Lemma}
\newtheorem{corollary}{Corollary}

\newtheorem{observation}{Observation}

\title{Ordered Leaf Attachment (OLA) Vectors can Identify Reticulation Events even in Multifurcated Trees}
\author{Alexey Markin$^{1,\ast}$, Tavis K. Anderson$^1$\\
\footnotesize\textit{$^{1}$~Virus and Prion Research Unit, National Animal Disease Center, USDA-ARS}
\\
\footnotesize\textit{*Email: alexey.markin@usda.gov}%
}

\date{}

\begin{document}

\maketitle

\begin{abstract}
    Recently, a new vector encoding, Ordered Leaf Attachment (OLA), was introduced that represents $n$-leaf phylogenetic trees as $n-1$ length integer vectors by recording the placement location of each leaf. Both encoding and decoding of trees run in linear time and depend on a fixed ordering of the leaves. Here, we investigate the connection between OLA vectors and the maximum acyclic agreement forest (MAAF) problem. A MAAF represents an optimal breakdown of $k$ trees into reticulation-free subtrees, with the roots of these subtrees representing reticulation events.
    We introduce a \emph{corrected OLA distance} index over OLA vectors of $k$ trees, which is easily computable in linear time. We prove that the corrected OLA distance corresponds to the size of a MAAF, given an optimal leaf ordering that minimizes that distance. Additionally, a MAAF can be easily reconstructed from optimal OLA vectors. We expand these results to multifurcated trees: we introduce an $O(kn \cdot m\log m)$ algorithm that optimally resolves a set of multifurcated trees given a leaf-ordering, where $m$ is the size of a largest multifurcation, and show that trees resolved via this algorithm also minimize the size of a MAAF. These results suggest a new approach to fast computation of phylogenetic networks and identification of reticulation events via random permutations of leaves. Additionally, in the case of microbial evolution, a natural ordering of leaves is often given by the sample collection date, which means that under mild assumptions, reticulation events can be identified in polynomial time on such datasets.
\end{abstract}

\section{Introduction} 
Identifying reticulation events and the construction of phylogenetic networks that integrate reticulate evolution is an important problem in evolutionary biology~\cite{Kong:2025biodiversity}. Reticulate evolution includes such evolutionary processes as recombination, hybridization, horizontal gene transfer, and reassortment~\cite{Huson:2011survey}. Developing methods that identify and quantify such events is critical to our understanding of evolution broadly and the impact of reticulate evolution in particular (cf.~\cite{Awadalla:2003recombination,Arnold:2006primates,Hibbins:2021tomatoes}). However, such development is compounded by the computational complexity of network inference problems~\cite{Bordewich:2007hybrid,Kong:2022net-classes,Markin:2019rf-net} and the challenge to distinguish between tree discrepancy due to errors/incomplete lineage sorting/convergence and discrepancy due to true reticulation~\cite{Solis:2017phylonetworks}.

Recently, three new methods for encoding phylogenetic trees as compact vectors were introduced, including the HOP-encoding~\cite{Chauve:2025HOP}, Phylo2Vec~\cite{Penn:2025phylo2vec}, and Ordered Leaf Attachment (OLA)~\cite{Richman:2025OLA}. All three methods encode binary rooted phylogenetic trees into $O(n)$-length vectors, where $n$ represents the number of leaves, and depend on a fixed ordering of the leaves/taxa. OLA and Phylo2Vec encodings induce natural distance metrics between phylogenetic trees over the same taxon set: the Hamming distance between the corresponding vectors. The HOP-distance definition is more involved as it relies on the notion of longest common subsequences. The OLA Hamming distance is the fastest to compute as it requires $O(n)$ time. Notably, Linz et al.~\cite{Linz:2025order} recently showed that for all three encodings, the corresponding distance, given an optimal leaf-ordering that minimizes that distance, is upper-bounded by the hybridization number of the trees. In the case of the HOP-distance, the minimum distance is equivalent to the hybridization number. Note that the hybridization number for two binary trees is the smallest number of reticulation events that is required to explain the topological difference between the trees, and is equivalent to the size of a maximum acyclic agreement forest (MAAF) minus 1 \cite{Baroni:2005bhybrid,Bordewich:2007hybrid}. For more than two trees, the relationship between the hybridization number and the size of a maximum acyclic agreement forest is more complex~\cite{vanIersel:2014three}; however, if we allow reticulation events to have more than two parents, the size of a MAAF minus 1 still represents the smallest number of reticulation events required to explain the topological differences between trees~\cite{vanIersel:2013hyb-multi}. Therefore, from now on, given $k > 1$ trees, we will refer to the size of MAAF minus 1 as the \emph{reticulation number} to distinguish it from the hybridization number for $k > 2$.

In this work, we explore the connection between OLA vectors and the maximum acyclic agreement forests for two or more trees. The OLA encoding is arguably the most convenient to work with of the three encodings described above. This is facilitated by simple linear-time encoding and decoding algorithms for OLA vectors and OLA's interpretability~\cite{Richman:2025OLA}. Given a set of $k$ rooted binary trees $\T$ over $n$ leaves and a fixed leaf-ordering, we define the Hamming OLA distance between these trees as the total number of indices where at least two vectors differ. We can show that this distance can be strictly smaller than the reticulation number for $\T$ (see Fig.~\ref{fig:correction}B). To address this, we define a corrected OLA distance, where, in addition to mismatched indices, we count certain consensus indices as mismatches as well. In particular, we count a leaf $y$ as a mismatch if it is placed above a previously mismatched node (see Fig.~\ref{fig:correction}). Note that the corrected OLA distance can still be computed in linear time. As one of the main results of the work, we show that under an optimal leaf-ordering, the corrected OLA distance is equivalent to the reticulation number. We demonstrate this result by exploring the connection between OLA vectors and Acyclic Agreement Forests (AAFs). Given an AAF of size $f$, we can define a leaf-ordering such that the corresponding OLA vectors have the corrected OLA distance (and Hamming distance) of at most $f - 1$. Similarly, given a fixed leaf ordering with the corrected OLA distance $o$, we can construct an AAF of size at most $o + 1$ in linear time.

Next, we expand these results to multifurcated trees. For a set of multifurcated trees, the reticulation number is defined as the smallest reticulation number among all possible binary resolutions of these trees. Given a fixed leaf-ordering, we propose a polynomial-time algorithm that optimally resolves the multifurcated input trees, minimizing the corrected OLA distance (Algorithm~\ref{alg:resolve}). We prove that given an optimal leaf ordering, the resulting resolved trees also have the smallest reticulation number among all possible tree resolutions. Our algorithm runs in $O(kn \cdot m\log m)$ time, where $m$ is the size of a largest multifurcation.

The above results imply that when a leaf-ordering is fixed, the corresponding (order-dependent) reticulation number can be computed in polynomial time (linear time for binary trees). This is particularly important for fast-evolving pathogens, such as RNA viruses, where the collection dates of the pathogens induce the natural ordering of the leaves. In such cases, assuming that ancestral genotypes are rarely sampled after the descendant genotypes, the reticulation events are straightforward to identify via OLA vectors. For other organisms, where the leaf-ordering is not readily available, we suggest the following sampling approach to upper-bound the reticulation number and find reasonable agreement forests: sample $X$ random leaf permutations and compute the minimum corrected OLA distance across the samples. Our proposed Algorithm~\ref{alg:resolve} also ensures that the influence of tree estimation error on the reticulation number and AAFs can be minimized via collapsing short branches and/or branches with low bootstrap support. Additionally, for fast-evolving pathogens, such as RNA viruses, Algorithm~\ref{alg:resolve} can be used on its own to optimally resolve polytomies. We implemented these algorithms in a tool called OLA-Net, available at \url{https://github.com/flu-crew/OLA-Net}.


\begin{figure}
    \centering
    \includegraphics[width=0.8\linewidth]{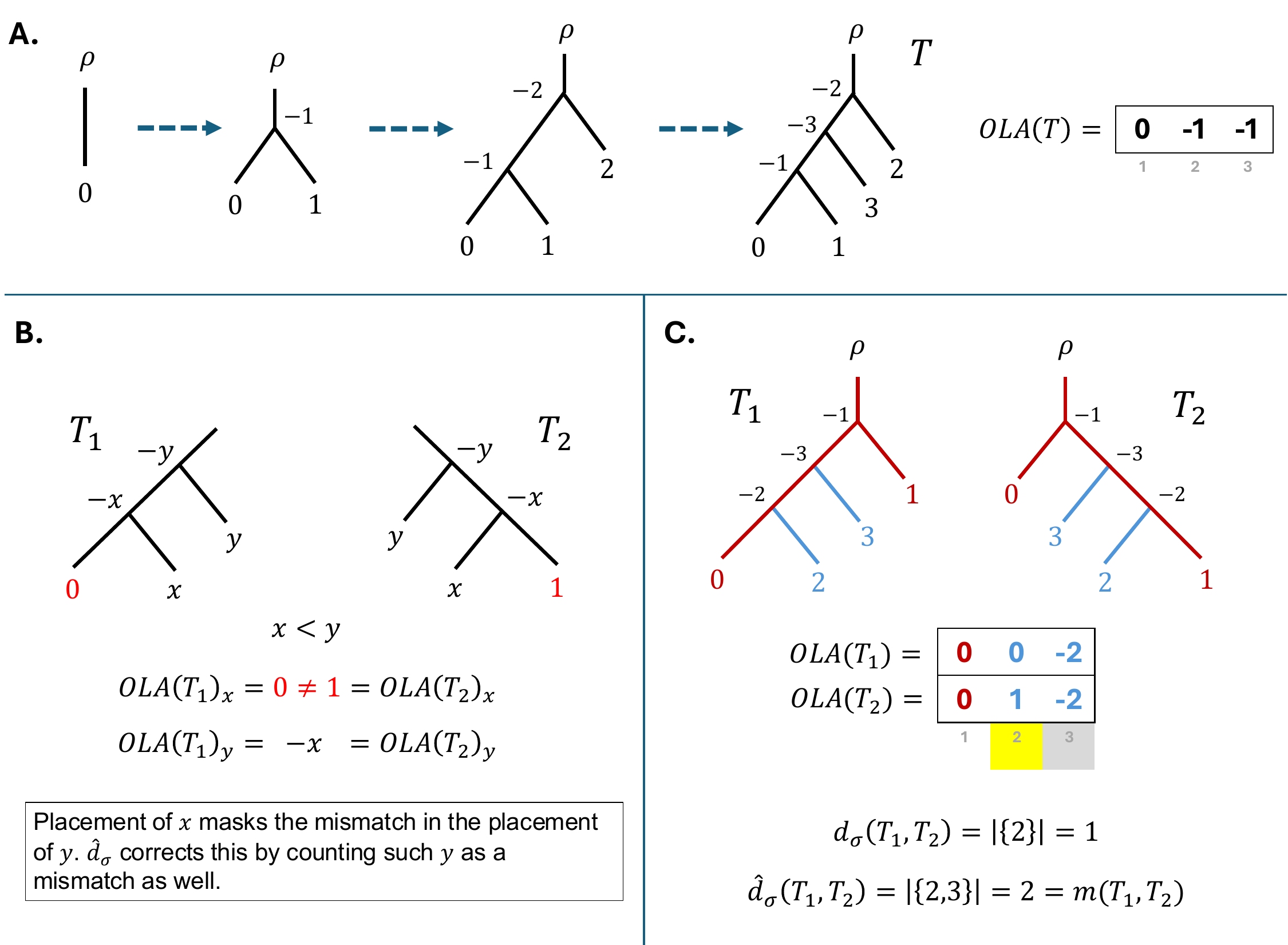}
    \caption{Illustrating the idea behind the corrected OLA distance. \textbf{A}: an example of how an OLA vector is constructed; the leaves are added in order, and the placement of each next leaf is recorded in the OLA vector. \textbf{B}: leaf $x$ precedes $y$ in the ordering, and is placed on the phylogenies first. The placement of $x$ causes a mismatch as it is placed above $0$ in $T_1$ and above $1$ in $T_2$. Although $y$ is essentially placed on the same edge as $x$, this placement does not cause a mismatch, as in both trees, $y$ is placed above node $-x$. In the corrected OLA distance, we count such $y$ as a mismatch. \textbf{C}: An example of two trees showing that the corrected OLA distance measures the reticulation number, whereas the original Hamming distance underestimates the reticulation number. In this example, $x = 2$ and $y = 3$.}
    \label{fig:correction}
\end{figure}

\section{Preliminaries}
A \emph{(phylogenetic) tree} over a leaf-set $L$ is a rooted and full binary tree $T = (V(T), E(T))$ with its leaves uniquely labeled (i.e., identified) by the elements of $L$. The trees are \emph{planted}, meaning that we add a single node above the root labeled by a unique label $\rho$. For convenience, we consider $\rho$ to be in $L$. The actual root of $T$ (the child of $\rho$) is denoted $r(T)$. 
For a subset $S \subset L$, $T(S)$ denotes the smallest subtree of $T$ that connects all nodes in $S$. The reduction $T|_S$ can be obtained from $T(S)$ by suppressing all non-labeled nodes with less than two children. Note that $T|_S$ or $T(S)$ are planted only if $\rho \in S$. For a tree $T$, we denote its set of labels by $L(T)$.

For a node $v$ in $T$ we denote its parent by $p(v)$ and its children, $ch(v)$, by $v_{(1)}$ and $v_{(2)}$ if such nodes exist. By $T_v$ we denote the subtree of $T$ rooted at $v$. The \emph{cluster} of $v$ is $cl(v) := L(T_v)$. The internal nodes of $T$ are denoted by $V_I(T) = V(T) \setminus L$.

In Section~\ref{sec:multi}, we allow trees to be non-binary or \emph{multifurcated}, i.e., each non-labeled node can have two or more children. We say that a binary tree $T'$ \emph{resolves} a non-binary tree $T$ if $L(T') = L(T)$ and there exists a set of edges $E \subset E(T')$ so that contracting these edges on $T'$ will result in the tree $T$.

\paragraph{Acyclic Agreement Forests.} Let $\F = \set{C_1, \ldots, C_k}$ be a forest with every component (tree) in $\F$ having a distinct set of labels (non-overlapping). We say that $\F$ \emph{agrees with} tree $T$ if

\begin{itemize}
    \item[(i)] For each $C_i \in \F$, $T|_{L(C_i)} \cong C_i$. That is, on the set of labels of $C_i$, $T$ is isomorphic to $C_i$.
    \item[(ii)] The trees $\set{T(L(C_i)) \mid C_i \in \F}$ are node-disjoint subtrees of $T$.
    \item[(iii)] $\displaystyle \bigcup_{C_i \in \F}L(C_i) = L(T)$.
\end{itemize}

\noindent For a set of trees $\mathcal{T}$ over label-set $L$, $\F$ is an \emph{agreement forest} for $\mathcal{T}$ if $\F$ agrees with every tree in $\mathcal{T}$.
An \emph{inheritance graph} $G(\T, \F)$ is a directed graph that has components $C_i \in \F$ as nodes. $G(\T, \F)$ has an edge $(C_i, C_j)$ if and only if there is a tree $T \in \T$ that contains a directed path from the root of $T(L(C_i))$ down to the root of $T(L(C_j))$. We say that $\F$ is an \emph{acyclic agreement forest} for $\T$ if $\F$ is an agreement forest for $\T$ and $G(\T, \F)$ does not contain directed cycles.

Finally, $\F$ is a \emph{maximum acyclic agreement forest} (MAAF) for $\T$ if there does not exist another acyclic agreement forest $\F'$ with fewer components than $\F$. MAAFs are strongly related to the \emph{hybridization number} problem. In particular, for two trees $T_1, T_2$ the hybridization number is $|\maaf(T_1, T_2)| - 1$. For a set of more than two trees, $\T$, the relationship is more complex, and $|\maaf(\T)| - 1$ is a lower bound on the hybridization number~\cite{vanIersel:2014three}. For any set of trees $\T$, we define $m(\T) := |\maaf(\T) - 1|$ to be the \emph{reticulation number} of $\T$. That is, the reticulation number is the smallest number of reticulation events required to explain the differences between the trees in $\T$ if we allow reticulations to have more than two parents.

\begin{figure}[!t]
    \centering
    \includegraphics[width=1\linewidth]{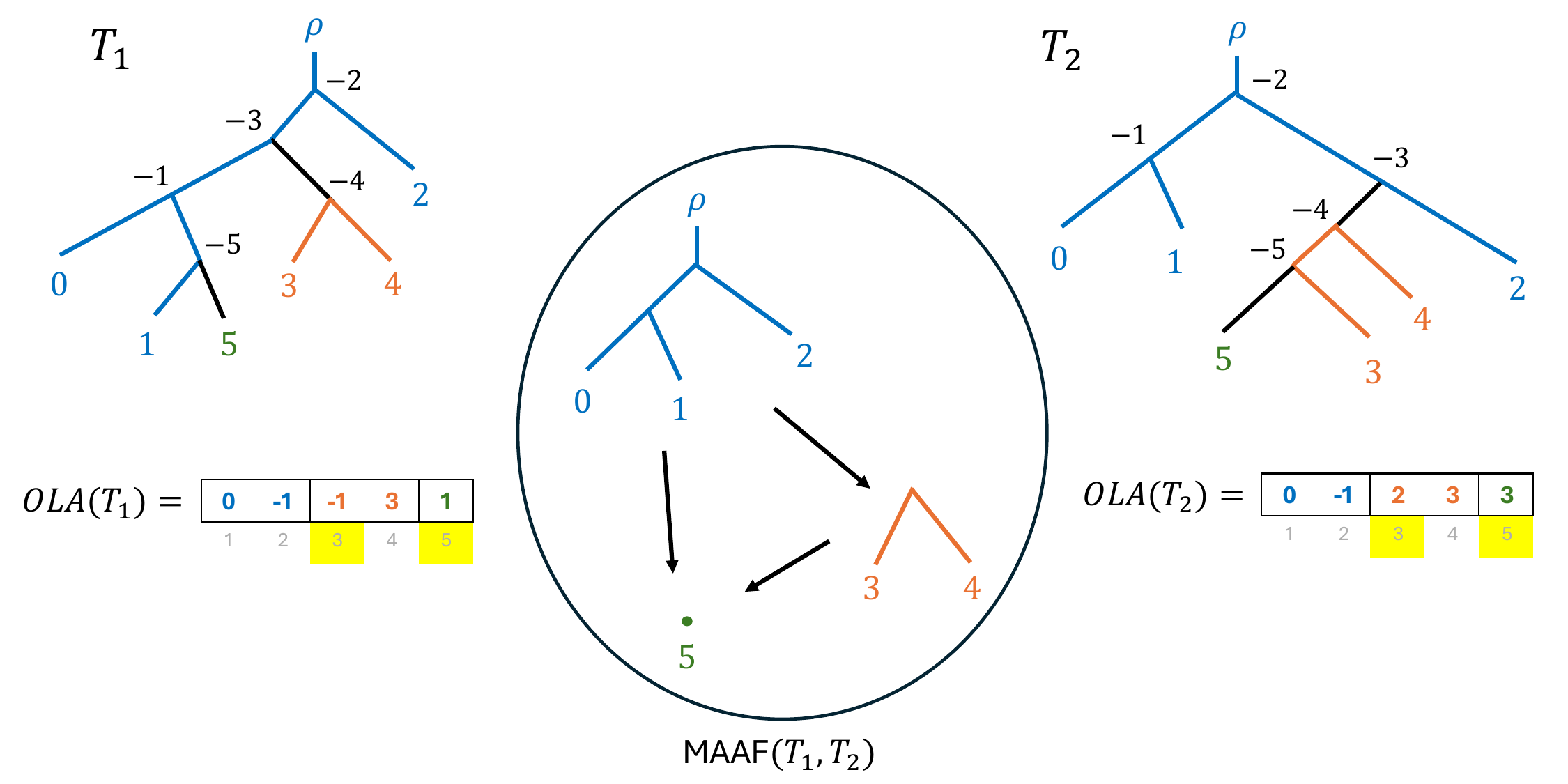}
    \caption{An example of a Maximum Acyclic Agreement Forest and optimal OLA-vectors for trees $T_1$ and $T_2$ over the leaf-set $\set{0, 1, \ldots, 5}$. The internal nodes of $T_1$ and $T_2$ are indexed by the OLA schema following the natural leaf-ordering $(0, 1, \ldots, 5)$. The optimal ordering of leaves follows the inheritance graph of $\F = \maaf(T_1, T_2)$ in the topological order. For each subtree $C$ of $\F$, only the entry corresponding to the first leaf of $C$ may differ between OLA-vectors; in this example it's leaves $3$ and $5$.}
    \label{fig:OLA-MAAF}
\end{figure}

\paragraph{Ordered Leaf Attachment.}

Let $(l_0, \ldots, l_{n-1})$ be a fixed ordering of leaves in label-set $L$, excluding $\rho$. In particular, $\sigma\colon L \setminus \set{\rho} \to \set{0, \ldots n-1}$ is a bijective mapping from a leaf to its 0-based index (i.e., $\sigma(l_i) = i$). For a tree $T$ over label-set $L$, we index the internal nodes of $T$ as follows. For each $v \in V(T)$, $\mu(v) := \min \set{\sigma(l) \mid l \in L(T_v)}$. Then, for each internal $v \in V(T) \setminus L$, $\sigma(v) :=  -\max\set{\mu(v_{(1)}), \mu(v_{(2)})}$. Here, $\sigma(v)$ is the index of an internal node $v$. That is, for a node $v$, we consider the smallest (earliest in the order) leaves within the two subtrees below $v$: let those leaves be indexed $i$ and $j$. If $i > j$, then $v$ gets assigned index $-i$ and otherwise index $-j$.

For $i \in \set{0, 1, \ldots, n-1}$, let $T^i$ be a shorthand for the tree $T|_{\set{l_0, \ldots, l_i}}$ that is a restriction of $T$ to the leaves indexed up to $i$. We define $v_i$ to be the sibling of $l_i$ in $T^i$. That is, $v_i$ is the node above which leaf $l_i$ is placed when we consider a step-by-step construction of $T$.  An OLA vector for tree $T$, $OLA(T,\sigma)$, is a vector of length $n-1$ such that for $i \in \set{1, \ldots, n-1}$
$$
    OLA(T, \sigma)_i := \sigma(v_i)
$$

Informally, the OLA vector can be seen as a process of building the tree in the order of its leaves. We start with a tree containing a single leaf $l_0$. Then, leaf $l_1$ gets added above $l_0$, and therefore the first entry in the vector is 0. The internal node that appears as a result of the addition of $l_1$ is labeled $-1$. Then, at each step $i > 1$, if a leaf $l_i$ is added above node $v_i$, then the index of $v_i$ is added to the vector, and the new internal node (i.e., the parent of $l_i$) is indexed as $-i$.

The vector has several important properties that make it very appealing for phylogenetic applications. First, the structure of the vector can be easily described as follows:
\[
 -(i-1) \le OLA(T, \sigma)_i \le  i - 1 
\]
for each $i \in \set{1, \ldots, n-1}$. Further, any integer vector that has this structure uniquely represents a phylogenetic tree. That is, Richman et al.~\cite{Richman:2025OLA} showed a bijection between the set of rooted binary trees over $n$ leaves and the set $\mathcal{C}_{n-1} := \set{(a_1, \ldots, a_{n-1}) \in \Z^{n-1} \mid -(i-1)\le a_i \le (i-1)}$. Importantly, both the encoding of a tree as a vector and the decoding of a tree from a vector can be performed in linear time.

Since all nodes of $T$ are uniquely indexed by $\sigma$, we sometimes refer to a node by its index (if an ordering $\sigma$ is fixed).

The OLA vectors induce a distance metric on the space of rooted binary phylogenetic trees: the Hamming distance between two OLA vectors. More formally, given a fixed ordering $\sigma$ and two trees $T_1$ and $T_2$, we define

\[
d_\sigma(T_1, T_2) := ||OLA(T_1, \sigma) - OLA(T_2, \sigma)||_0,
\]
where $||\cdot||_0$ is the $L_0$ norm.

\section{Correcting the OLA distance}
Under some orderings, the Hamming OLA distance, $d_{\sigma}(T, T')$, can be lower than the corresponding reticulation/hybridization number for $T$ and $T'$ (cf. Fig.~\ref{fig:correction}C). To correct for this, we define an adjusted distance measure $\hat{d}_\sigma$ that can still be easily computed in linear time. In particular, given two OLA vectors, we progressively build a set of ``mismatched'' indices $M$. Starting with an empty set, we iterate over $i \in \set{1, \ldots, n-1}$. We add $i$ to $M$ if 
\begin{itemize}
\item $OLA(T, \sigma)_i \ne OLA(T', \sigma)_i$, or 
\item $OLA(T, \sigma)_i = OLA(T', \sigma)_i = -j \text{, s.t. } j \in M$.
\end{itemize}
That is, in addition to regular mismatched entries of the OLA vectors, we count as mismatches those leaves that were added on top of the placements of already mismatched leaves. Intuitively, if we know that the original placement was a mismatch, then the placement one spot over it must also be a mismatch (Fig.~\ref{fig:correction}A). Then,
\[
    \hat{d}_{\sigma}(T, T') := |M|.
\]

We extend this \emph{corrected OLA distance} definition from two trees to any number of trees $k > 1$. For a set of trees $\T$ and an ordering $\sigma$, let $M$ be the set of mismatched indices, such that for each $i \in M$ either $\exists T, T' \in T \text{ s.t. } OLA(T, \sigma)_i \ne OLA(T', \sigma)_i$ or $\forall T \in \T : OLA(T, \sigma)_i = -j \text{ with } j \in M$. Then, $\hat{d}_{\sigma}(\T) := |M|$.

\section{Relationship between OLA vectors and acyclic agreement forests}\label{sec:maaf-bin}

In this section, we state and prove one of our main results, that the size of the maximum acyclic agreement forest for a set of trees equals their minimum (corrected) OLA distance plus 1. That is, the minimum corrected OLA distance is the reticulation number.

\begin{theorem}\label{thm:OLAk}
    Let $\T$ be a set of $k > 1$ trees over the same label-set $L$. Let $\sigma^*$ be the optimal leaf-ordering for these trees that minimizes the corrected OLA distance: $\displaystyle\sigma^* = \argmin_{\sigma} \hat{d}_{\sigma}(\T)$. Then,
    \[
        \hat{d}_{\sigma^*}(\T) = m(\T) = |\maaf(\T)| - 1.
    \]
\end{theorem}
\medskip
The rest of the section is dedicated to the proof of Theorem~\ref{thm:OLAk}. \smallskip

\noindent $\mathbf{(\le)}$
We begin by showing that $\hat{d}_{\sigma^*}(\T) \le |\maaf(\T)| - 1$. Let $\F$ be an acyclic agreement forest for $\T$; we are going to create an ordering $\sigma$ such that $\hat{d}_{\sigma}(\T) \le |\F| - 1$. Since $G(\T, \F)$ is acyclic, let $O = (C_1, \ldots, C_k)$ be a topological ordering of trees in $\F$ according to $G(\T, \F)$. We then construct $\sigma$ to be an ordering that puts $L(C_1)$ leaves first, then $L(C_2)$, $L(C_3)$, etc. That is, $\sigma$ is any ordering that satisfies the following condition
\[
\forall l \in L(C_i) \text{ and } l' \in L(C_j) \text{ with } i < j : 0 \le \sigma(l) < \sigma(l') \le n - 1.
\]
Consider the first $|L(C_1)| - 1$ entries of each vector $OLA(T, \sigma)$. Since $C_1$ is an agreement subtree for all $T \in \T$, the first $|L(C_1)| - 1$ entries must be identical for all trees $T$. The next entry corresponds to the positioning of the subtree $C_2$ relative to $C_1$ in each tree $T$, and therefore, this position in the OLA vectors could differ. However, the following $|L(C_2)| - 1$ entries must match, and so on. More formally, we prove the following lemma by induction.

\begin{lemma}\label{lem:C_i}
    Consider $C_i \in \F$ for some $1 \le i \le k$, let $l_s$ be the first leaf in $L(C_i)$ according to $\sigma$ and let $l_j$ be \emph{some} leaf in $L(C_i)$. Recall that $T^j$ is a shorthand for $T|_{\set{l_0, \ldots, l_j}}$. Then, for any $T \in \T$ there exists a node $x \in V(T^j)$ such that $(T^j)_x = C_i|_{\set{l_s, \ldots, l_j}}$ and the nodes of $(T^j)_x$ are indexed identically across all $T \in \T$. In particular, for trees $T, T' \in \T$ and corresponding subtree roots $x, x'$, let $v \in V((T^j)_x)$ and $v' \in V((T'^j)_{x'})$ such that $L((T^j)_v) = L((T'^j)_{v'})$; then, $\sigma(v) = \sigma(v')$.
\end{lemma}
\begin{proof} Base case: $l_j = l_s$. In this case, $x = l_j$ with $(T^j)_x$ being a trivial subtree with one leaf. Further, $l_j$ is consistently indexed by $j$ across all trees $T \in \T$.

Consider now some $j > s$ and we assume that the statement holds for all $j' \in \set{s, \ldots, j-1}$. Let $x_{j-1}$ denote the root of subtree $C_i|_{\set{\l_s, \ldots, l_{j-1}}}$ in $T^{j-1}$.
Consider now node $v_j$ that is a sibling of $l_j$ in $T^j$. Assume that $v_j$ is not in $(T^{j-1})_{x_{j-1}}$. Then, $v_j$ must have at least one leaf $l' \not \in L(C_i)$ below it in $T^{j-1}$. If $l' \in L(C')$, then either $T^{j}(L(C_i))$ overlaps with $T^{j}(L(C'))$ (which is a contradiction) or the root of $T^{j}(L(C_i))$ is above the root of $T^{j}(L(C'))$, since both roots must be located on the path from $l'$ to $\rho$. The latter case implies that $C_i$ is above $C'$ in $G(\T, \F)$, which is also a contradiction, since $C_i$ must follow $C'$ in the topological ordering. That is, $v_j$ must be a node in $(T^{j-1})_{x_{j-1}}$.

If $v_j = x_{j-1}$, then the parent of $l_j$ becomes the new root of the subtree $C_i|_{\set{l_s, \ldots, l_j}}$. Otherwise, the parent of $l_j$ is a new internal node in that subtree that is still rooted at $x_{j-1}$. In both cases, the parent al $l_j$ gets consistently indexed as $-j$ across all $T \in \T$.
\end{proof}

\begin{corollary}\label{cor:C_i}
    Let $C_i \in \F$ for some $1 \le i \le k$ and let $l_j \in L(C_i)$ be a leaf that is not the first leaf in $C_i$ (i.e., $l_{j-1} \in L(C_i)$). Then, $OLA(T, \sigma)_j = OLA(T', \sigma)_j$ for all $T, T' \in \T$.
\end{corollary}
\begin{proof}
    Let $v_j$ be the sibling of $l_j$ in $T^j$. As was shown above, $v_j$ must be identically indexed across all $T \in \T$. Therefore, $OLA(T, \sigma)_j = OLA(T', \sigma)_j$ for all $T, T' \in \T$.
\end{proof}

Corollary~\ref{cor:C_i} implies that the only mismatched indices across different OLA vectors could appear once per component $C_i \in \F$ and correspond to the first leaf in each component. Since the first leaf of $C_1$ (i.e., $l_0$) does not appear in OLA vectors, there could be at most $|\F| - 1$ mismatched indices. The only thing left to show is that when we compute the corrected OLA distance, we do not add any ``corrected'' mismatches.

\begin{lemma}\label{lem:upper-bound}
    For the chosen ordering $\sigma$, $\hat{d}_{\sigma}(\T) = d_{\sigma}(\T) \le m(\T) = |\F| - 1$.
\end{lemma}
\begin{proof}
    Let $M = \set{j > 0 \mid l_j \in L(C_i)\ \&\ \nexists l_x \in L(C_i) \text{ with } x < j}$ be the set of the first indices in each of the components $C_i$ for $i > 0$. Further, let $M^{-} := \set{-j \mid j \in M}$.
    Then consider the set $I = \set{1, \ldots, n-1} \setminus M$. We claim that there is \emph{no} such $j \in I$ for which $\exists T \in \T$ such that $OLA(T, \sigma)_j \in M^{-}$.
    Consider any such $j$ and the corresponding component $C_i$ ($l_j \in C_i$). By Lemma~\ref{lem:C_i}, $l_j$ must have been placed within subtree $(T^{j})_x$ that consisted only of nodes in $C_i$. If $l_s$ is the first leaf in $C_i$, then $x$ must be below the node indexed by $-s$ and therefore $l_j$ could not have been placed above any node indexed by some $y \in M^{-}$.
\end{proof}

\medskip
\noindent$(\ge)$ Next, we demonstrate that $\hat{d}_{\sigma^*}(\T) \ge |\maaf(\T)| - 1$. Consider a fixed ordering of leaves $\sigma$. It is sufficient to show that we can construct an AAF $\F$ for $\T$ such that $|\F| - 1 \le \hat{d}_{\sigma}(\T)$.

Let $M = \set{m_1, \ldots, m_k}$ be the mismatched indices including the corrected mismatches. That is, for each $i \in M$ either $\exists T, T' \in \T$ s.t. $OLA(T, \sigma)_{i} \ne OLA(T', \sigma)_{i}$ or $OLA(T, \sigma)_{i} \in \set{-m \mid m \in M}$. Additionally, we set $m_0 = 0$. All other indices $I_c = \set{1, \ldots, n-1} \setminus \set{m_1, \ldots, m_k}$ must be ``consensus'' indices for which all OLA vectors agree. Next, we progressively define a subset of leaves $L_i$ for each $i \in \set{0, \ldots, k}$. We start with $L_i = \set{l_{m_i}}$ for $i > 0$ and $L_0 = \set{0, \rho}$. We define the span of $L_i$ to be

\[span(L_i) := \bigg( \bigcup_{l_j \in L_i \setminus \set{m_i, \rho}}\set{-j, j}\bigg) \; \cup \set{m_i}\]

Then, for each consensus index $j \in I_c$, in order, we add $l_j$ to $L_i$ if $OLA(T, \sigma)_j \in span(L_i)$. Note that such $L_i$ must exist and be unique. As a result, sets $\set{L_i}$ partition the label-set $L$. Fig.~\ref{fig:AAF-from-OLA} shows an example of such partitioning and a resulting AAF.

\begin{figure}[!t]
    \centering
    \includegraphics[width=0.9\linewidth]{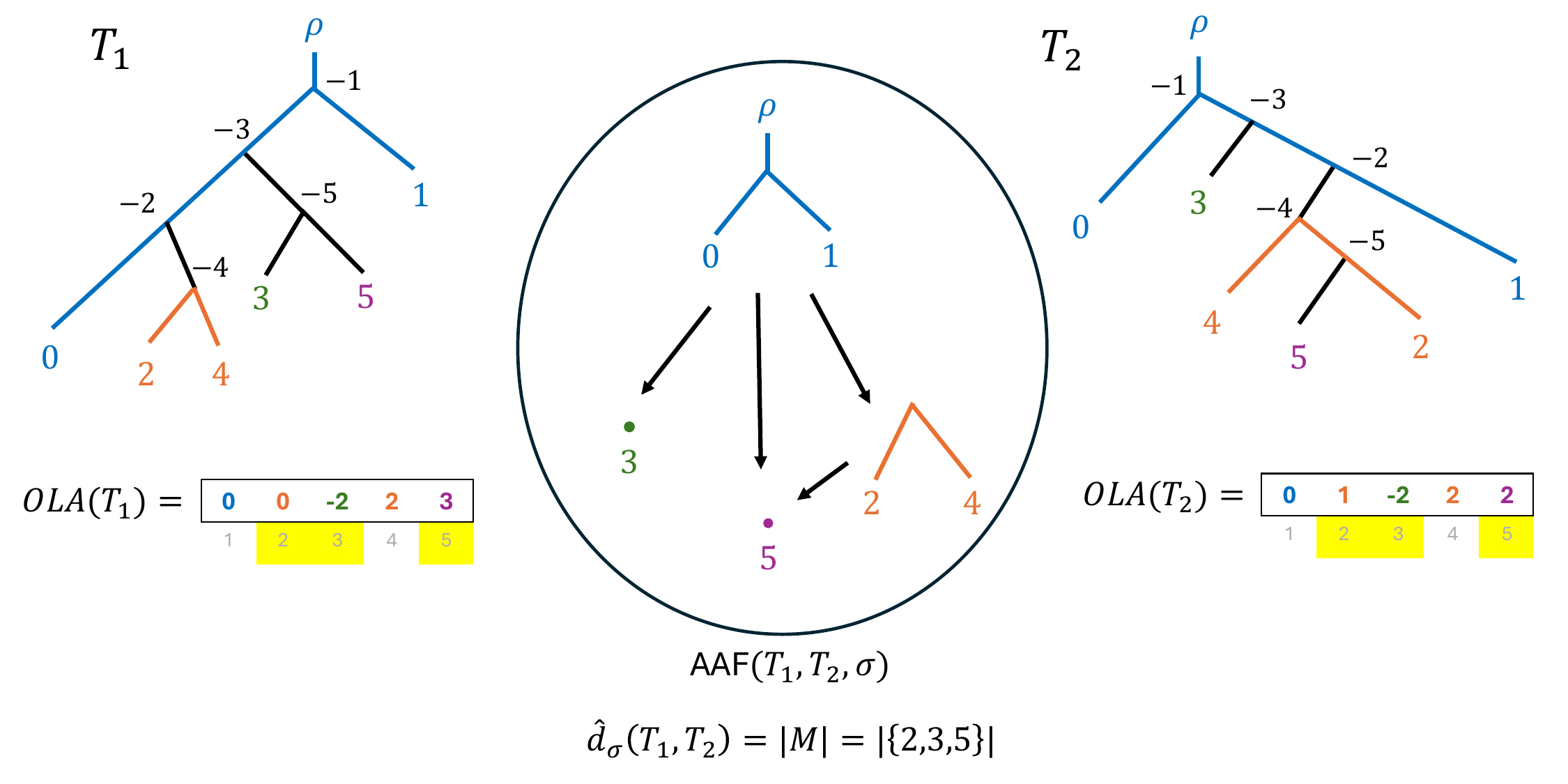}
    \caption{An example of constructing an AAF from the OLA vectors of two trees, $T_1$ and $T_2$. This example illustrates the proof of the lower-bound in Theorem~\ref{thm:OLAk}. We start with $L_0 = \set{\rho, 0, 1}$. Then, leaf $2$ is a mismatch and it creates $L_1 = \set{2}$. As leaf $3$ is placed above $-2$ in both trees with $-2 \in M^-$, it is also a mismatch and results in $L_2 = \set{3}$. Leaf $4$ has a consensus placement and it's in $span(L_1)$; hence, we add $4$ to $L_1$. Finally, leaf $5$ is a mismatch, creating $L_3 = \set{5}$.}
    \label{fig:AAF-from-OLA}
\end{figure}

\begin{lemma}\label{lem:L_i}
    Nodes in $T|_{L_i}$ (except for $\rho$) are bijectively indexed by $span(L_i)$ for all $T \in \T$. Further, for nodes $x \in V(T|_{L_i})$ and $x' \in V(T'|_{L_i})$ with $\sigma(x) = \sigma(x')$ the leaf-sets underneath these nodes are identical: $L_i \cap L(T_x) = L_i \cap L(T'_{x'})$.
\end{lemma}
\begin{proof}
    We prove this by induction on partial trees $T^j$. The base case for $j=0$ holds trivially as $l_0$ is indexed by $0$ and $span(L_0) = \set{0}$ at the start.

    Consider now some $j > 0$; if $l_j = m_i \in M$. Then $span(L_i) = \set{m_i}$ and the statement holds similarly to the base case. Assume now $j \in I_c$ and is added to the set $L_i$. Consider the node $v_j$ that is a sibling of $l_j$ in $T^j$. Let $p$ denote the parent of $l_j$ and $v_j$ in $T^j$. Recall that $\sigma(v_j) \in span(L_i)$ by construction, and hence $v_j$ must have leaves from $L_i$ underneath it. Therefore, both children of $p$ contain a leaf from $L_i$ under them and hence $p \in V(T^j|_{L_i})$. Further, $\sigma(p) = -j$ and $-j \in span(L_i \cup \set{l_j})$ which confirms the bijection. Next, by the induction hypothesis, nodes $v_j$ have the same clusters $L_i \cap (T^j)_{v_j}$ across all trees $T \in \T$. As we add node $p$ directly above $v_j$, its cluster is $(L_i \cap (T^j)_{v_j}) \cup \set{l_j}$ in $T^j|_{L_i}$ (across all $T$), and clusters of all nodes $x \in V(T^j|_{L_i})$ above $p$ get expanded by one leaf.
\end{proof}

\begin{corollary}\label{cor:L_i}
    The forest $\F = \set{T|_{L_i} \mid i \in \set{0, \ldots, k}}$ for some $T \in \T$ is an agreement forest for $\T$.
\end{corollary}
\begin{proof}
    We need to show that (i) $T|_{L_i} = T'|_{L_i}$ for all $T, T' \in \T$ and all $i$, and (ii) $\set{T(L_i)}$ are node-disjoint subtrees of any $T \in \T$.

    (i) follows directly from Lemma~\ref{lem:L_i} since trees $T|_{L_i}$ and $T'|_{L_i}$ must have the identical set of clusters.

    (ii) we can prove by induction over partial trees $T^j$. The base case for $j=0$ clearly holds, as only $L_0$ is non-empty. Consider now any $j > 0$. If $l_j = m_i \in M$, then a new set $L_i$ becomes non-empty, which contains only one leaf and $T^j(L_i)$ cannot overlap with any other subtree. Now assume that $j \in I_c$ and is added to the set $L_i$. Consider node $v_j$ that is a sibling of $l_j$ in $T^j$. By definition of $L_i$, $\sigma(v_j) \in span(L_i)$ and hence $v_j \in V(T^{j-1}(L_i))$ by Lemma~\ref{lem:L_i}. As we add $l_j$ on top of $v_j$, the subtree induced by $L_i$ grows by one internal node and one leaf. The parent of $l_j$ in $T^j$ may be a part of another subtree $T(L_x)$ ($x \ne i$) if and only if $v_j$ was a part of that subtree, which contradicts the induction hypothesis.
\end{proof}

\begin{lemma}\label{lem:acyclic}
    For some $T \in \T$, let $r_i$ be the root of subtree $T(L_i)$. Then, $T_{r_i}$ cannot contain any $l_j$ such that $j < m_i$.
\end{lemma}
\begin{proof}
    It is sufficient to show that $T_{r_i}$ does not contain any $m_x \in M$ with $x < i$. Then, the statement will follow since subtrees $T|_{L_i}$ are node-disjoint. We prove this by induction over partial trees $T^j$. The base case of $j = 0$ clearly holds. For any $j > 0$, if $l_j = m_i \in M$, then $T^j(L_i)$ is trivial and the statement holds. Consider now $j \in I_c$ that is added to set $L_i$. As we showed above, $l_j$ is added adjacently to a node in $T^{j-1}(L_i)$ and therefore the cluster of $r(T(L_i))$ only grows by one and cannot include any $m_x$ with $x < i$.
\end{proof}

\begin{corollary}\label{cor:acyclic}
    The agreement forest $\F = \set{T|_{L_i} \mid i \in \set{0, \ldots, k}}$ is acyclic.
\end{corollary}
\begin{proof}
    Consider any $T(L_i)$ and $T(L_j)$ with $i > j$. For the root of $T(L_i)$, $r_i$, to be above of $r_j$, we need $T(L_j)$ to be nested within $T(L_i)$ which is prohibited by Lemma~\ref{lem:acyclic}. Therefore, $G(\T, \F)$ can never contain an edge from $T(L_i)$ to $T(L_j)$. Therefore, $(T_{L_i})_i$ is a topological ordering of $G(\T, \F)$, which means $G(\T, \F)$ is acyclic. 
\end{proof}

That is, we constructed an AAF $\F$ such that $|\F| - 1 = k = \hat{d}_{\sigma}(\T)$.

\section{Fast estimation of the reticulation number and phylogenetic network construction}
As we showed above, under an optimal ordering, the corrected OLA distance is equivalent to the reticulation number for a set of trees. This observation suggests a simple sampling algorithm that can upper-bound a reticulation number. Given a set of trees $\T$ over the same leaf-set, draw $X$ random permutations of the leaves and compute the minimum corrected OLA distance across these permutations. As OLA distances can be computed in linear time, this yields a very fast estimation algorithm for the reticulation number, and consequent construction of AAFs and phylogenetic networks. Additionally, straightforward parallelization of this algorithm means that a large value for $X$ can be chosen to yield better results.

Crucially, for some phylogenetic datasets, a natural ordering of leaves is defined by the order of sequence collection. This is particularly relevant for fast-evolving pathogens such as influenza A virus, SARS-CoV-2, and many other viruses and some bacteria. In this case, no random permutations are required, and the reticulation number can be estimated in linear time even on very large datasets with millions of strains.

For both scenarios, OLA vectors can be easily translated into an acyclic agreement forest and a phylogenetic network as described in the previous section. Thus, this approach yields a fast way to construct large phylogenetic networks.

\paragraph{Addressing topological errors.} One caveat with the estimation of the reticulation number and consequent network construction is that it can be very sensitive to even small topological errors in the trees~\cite{Markin:2022rf-net2}. One approach to reduce the influence of topological errors is to collapse all branches below a certain length threshold or below a fixed bootstrap support threshold. In this case, we need to extend the definition of OLA vectors to multifurcated trees, while maintaining the ability of OLA vectors to compute the reticulation number. In the next section, we show that this can be done.

\begin{algorithm}[!t]
    {\fontsize{9}{9}\selectfont
    \begin{algorithmic}[1]
        \State \textbf{Input:} Tree $T$ and ordering $\sigma$.
        \State \textbf{Output:} An OLA vector and a ``multi'' vector that notes leaves that contribute to multifurcations.
        \State Initialize $\texttt{OLA}, \texttt{multi}$ as vectors of size $n-1$
        \For{each node $v$ in postroder of $T$:}
            \If{$v$ is a leaf}
                \State $v.min = \sigma(v)$
            \Else
                \State $v.min = \min(\set{c.min \mid c \in ch(v)})$
                \State $\sigma(v) = -(\text{second smallest }c.min \text{ among the children})$
            \EndIf
        \EndFor

        \For{leaf $l_i$, $i>0$, in the reverse order of $\sigma$}
            \If{$p(l_i)$ is a multifurcation}
                \State $\texttt{OLA}_{i} = \sigma(p(l_i))$; $\texttt{multi}_i = True$
            \Else
                \State $\texttt{OLA}_{i} = \sigma(l_i.sibling)$; $\texttt{multi}_i = False$
            \EndIf
            \State Remove $l_i$ from $T$ and suppress a potential unifurcation at $p(l_i)$
        \EndFor
        \State \Return \texttt{OLA}, \texttt{multi}
    \end{algorithmic}}
    \caption{Tree preprocessing. We first assign indices to the internal nodes in a similar fashion to the original OLA algorithm, and we create a map (OLA vector) of where each leaf was added. For a leaf that creates a multifurcation, we record that multifurcated node as the placement node and mark it as a multifurcation in a ``multi'' vector.}
    \label{alg:preproc}
\end{algorithm}

\section{Extending OLA vectors to optimally resolve multifurcations}\label{sec:multi}
Given a set of not necessarily fully-resolved trees $\T = (T_1, \ldots, T_k)$ and a leaf ordering $\sigma$, we show how to resolve these trees while minimizing disagreement between them.

First, we build initial OLA vectors and identify the leaves, the addition of which creates multifurcations. We create vectors $\OLA(T)$ and $\multi(T)$ such that for $1 \le i \le n-1$,
\[
    \OLA(T)_i =
    \begin{cases}
        \sigma(v_i) \text{ if $p(l_i)$ is binary in $T^i$}\\
        \sigma(p(l_i)) \text{ otherwise}
    \end{cases}\;\;\;\;\;\;
    \multi(T)_i = \neg (\text{$p(l_i)$ is binary in $T^i$})
\]
This step can be performed in linear time (Algorithm~\ref{alg:preproc}). Next, we resolve the multifurcations across all trees $T \in \T$ by sequentially building resolved trees $T'$. For each multifurcated node in $T$, we maintain a subtree $B$ that represents the resolution of that multifurcation (i.e., collapsing all internal edges that do not lead to leaves in that subtree will restore the original multifurcation). Note that $B$ must be fully binary, and its leaves are the original children of the multifurcated node (cf. Fig.~\ref{fig:alg2}). The node indices of such a subtree are stored in a vector $\mset(T)$. That is, for an internal multifurcated node indexed $x$, the set $\mset(T)_x$ contains the indices of nodes in the subtree $B_x$ of $T'$ that resolves this multifurcation. Additionally, $\mroot(T)_x$ is the index of the root of $B_x$. Then, for every leaf index $i \ge 1$:
\begin{itemize}
    \item[(i)] We first add $l_i$ to all trees $T'$ where it does not cause a multifurcation (i.e., there is a unique placement). Let $P$ denote the set of placement indices of $l_i$ onto those trees. Note that if $|P| = 1$, then we have a potential consensus placement index $p \in P$.
    
    \item[(ii)] For each tree $T'$ where the addition of $l_i$ creates a multifurcation, let $x$ denote the index of that multifurcated node in the original tree $T$. If a consensus placement $p$ exists, then we check if $p \in \mset(T)_x$. If for one of the trees, $p$ is not in that set, we cannot avoid a mismatch at index $i$, and we add $l_i$ above the subtree-root, $\mroot(T)_x$.
    
    \item[(iii)] If $l_i$ creates a multifurcation across all trees $T$ (i.e., $P = \emptyset$), then we look at the intersection of the sets $\mset(T)_x$ across all $T \in \T$ (and exclude the parents of mismatched indices in $M^-$ from that intersection). If such an intersection is not empty, we can perform a consensus placement of $l_i$ across all trees.
\end{itemize}

The full algorithm for this approach is described in Algorithm~\ref{alg:resolve} and an example is given in Fig.~\ref{fig:alg2}.

\begin{algorithm}[!t]
    {\fontsize{9}{9}\selectfont
    \begin{algorithmic}[1]
        \State \textbf{Input:} Trees $\T$ and ordering $\sigma$.
        \State \textbf{Output:} Resolved tree $T'$ and OLA vector for each $T \in \T$.
        \State Initialize $T' = (\rho, l_0)$ for each $T \in \T$
        \State Get $\texttt{OLA}(T)$ and $\texttt{multi}(T)$ vectors via preprocessing (Alg.~\ref{alg:preproc}) for each $T$
        \State Initialize set $M^- = \emptyset$ // set of (negated) mismatched indices
        \State Initialize vectors $\texttt{m\_set}(T)$ (empty sets) and $\texttt{m\_root}(T)$ (integers) indexed from $-(n-1)$ to $n-1$
        \State // \textit{$\texttt{m\_set}(T)_i$ is a set of indices of a subtree representing a partially resolved multifurcation at node $i$ of $T$}
        \State // \textit{$\texttt{m\_root}(T)_i$ is the root of that subtree. We initialize $\mroot(T)_i = i$}
        \State Initialize integer vector $\leafof(T)$ indexed from $-(n-1)$ to $n-1$
        \State // \textit{$\leafof(T)_i = x$ if node $i$ is a leaf of a subtree that resolves the multifurcation at $x$ in tree $T$}
        \State // \textit{If $i$ is not a leaf of any such subtree, $\leafof(T)_i = 0$}
        \For{$i$ in $1\ldots n-1$}
            
            \State $\texttt{placements} = \set{\mroot(T)_{\texttt{OLA}(T)_i}) \mid \forall T \in \T \text{ if not } \texttt{multi}(T)_i}$ // {\it all resolved placements at this index}
            \State $\texttt{unresolved} = \set{T' \mid \forall T \in \T \text{ s.t. } \texttt{multi}(T)_i}$ // {\it all unresolved trees at this index}
            \State $\texttt{resolved} = \set{T' \mid \forall T\in \T} \setminus \texttt{unresolved}$
            \For{$T' \in \texttt{resolved}$}
                \State Add $l_i$ to $T'$ above node indexed $x=\mroot(T)_{\texttt{OLA}(T)_i}$ and set $\texttt{m\_set}(T)_{-i} = \set{i, -i, x}$; $\texttt{m\_root}(T)_{-i} = -i$
                \State If $y = \leafof(T)_x < 0:$ replace $x$ with $-i$ in $\mset(T)_{y}$ and set $\leafof(T)_{-i} = y$
                \State Set $\leafof(T)_i = \leafof(T)_x = -i$
            \EndFor
            \If{$|\texttt{placements}| > 1$}
                \State $p = \texttt{null}$; Add $-i$ to $M^-$
            \ElsIf{$|\texttt{placements}| = 1$}
                \State $p = \texttt{placements[0]}$
            \Else
                \State $\displaystyle I = \bigcap_{T \in \T}\texttt{m\_set}(T)_{x_T} \setminus M^-,$\ where $x_T = \texttt{OLA}(T)_i$ // \textit{common nodes across relevant multifurcations}\label{ln:inter}
                \If{$|I| = 0$}
                    \State $p = \texttt{null}$; Add $-i$ to $M^-$
                \Else
                    \State $p = \argmax_j\set{|j| \mid j \in I}$ // \textit{an index in $I$ with max absolute value}
                \EndIf
            \EndIf
            \For{$T' \in \texttt{unresolved}$}
                \State $x=\texttt{OLA}(T)_i$ // \textit{index of the multifurcated node}
                \If{$p \in \texttt{m\_set}(T)_x$}
                    \State $p_T = p$
                \Else
                    \State $p_T = \texttt{m\_root}(T)_x$; Add $-i$ to $M^-$
                \EndIf
                \State Add $l_i$ to $T'$ on top of node indexed $p_T$; Add $i$ and $-i$ to $\texttt{m\_set}(T)_x$; Set $\leafof(T)_{i} = x$
                \State If $p_T = \texttt{m\_root}(T)_x:$ set $\texttt{m\_root}(T)_x = -i$ // \textit{update the root}
                \State If $y = \leafof(T)_{p_T} < 0\ (x \ne y):$ replace $p_T$ with $-i$ in $\mset(T)_{y}$
                \State \;\;\; and set $\leafof(T)_{-i} = y, \leafof(T)_{p_T} = 0$
            \EndFor
        \EndFor
        \State Generate an OLA-vector for each resolved $T'$: Return resolved trees and vectors
    \end{algorithmic}}
    \caption{Multifurcation resolution. We progressively build resolved trees $T'$, at each index $i$ deciding if there is a way to place the leaf $l_i$ that is consistent across all trees.}
    \label{alg:resolve}
\end{algorithm}

\begin{figure}
    \centering
    \includegraphics[width=0.9\linewidth]{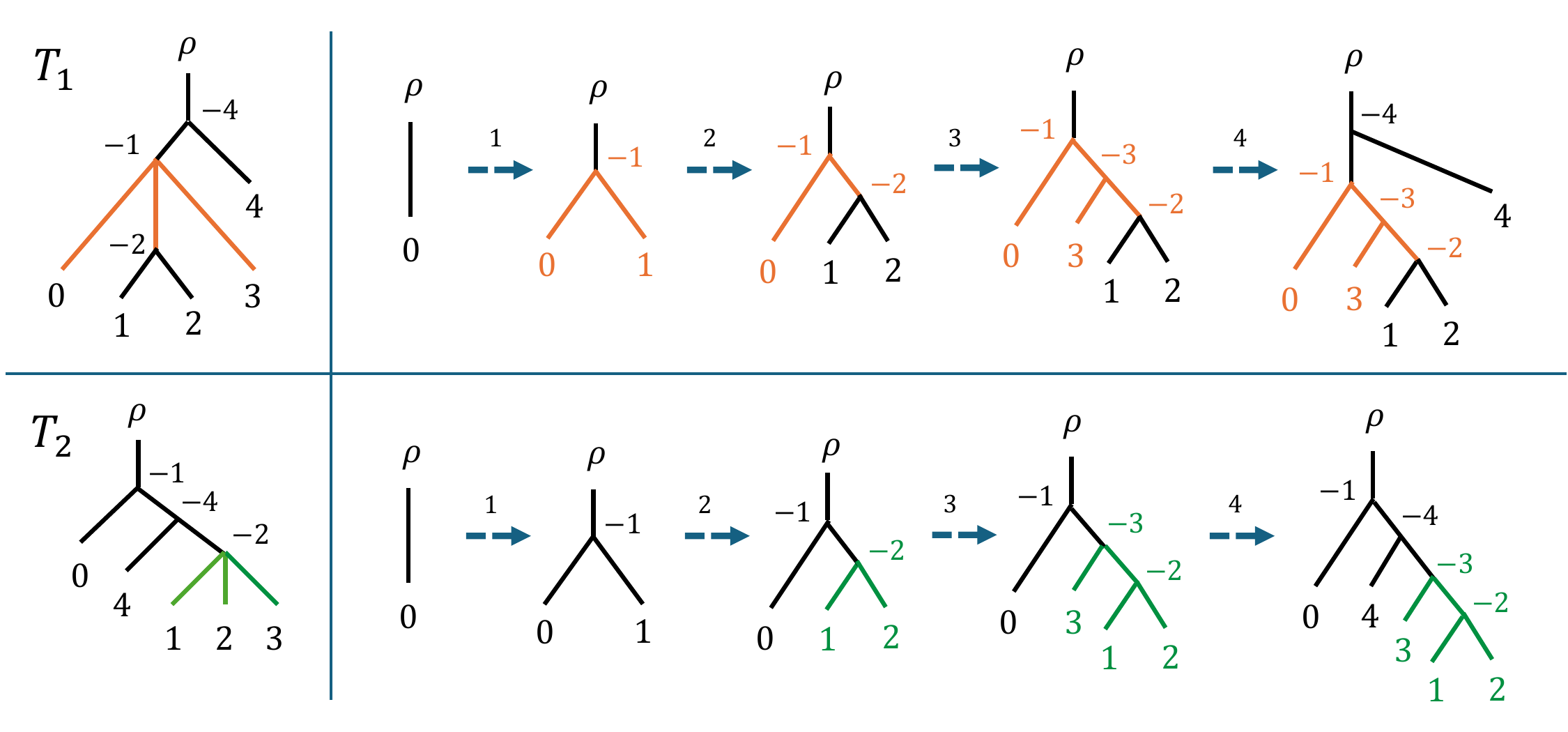}
    \caption{An example of tree resolution performed by Algorithm~\ref{alg:resolve} on a set of trees $\set{T_1, T_2}$. First, the nodes of both trees are indexed via Algorithm~\ref{alg:preproc}; $T_1$ has a multifurcation at node $-1$ and $T_2$ has a multifurcation at node $-2$. As we construct resolved trees, we keep track of all the nodes that are involved in the resolution of the multifurcations (highlighted in orange for $T_1$ and in green for $T_2$). Up to leaf $2$, all placements are binary and the partial trees are identical. For the non-binary leaf 3, we check the intersection of $\mset(T_1)_{-1} = \set{0,-1,-2}$ and $\mset(T_2)_{-2} = \set{1,-2,2}$, and we place $3$ on the common node $-2$. Finally, the placement of $4$ is binary in both trees.}
    \label{fig:alg2}
\end{figure}

\begin{lemma}\label{lem:mset-resolve}
    A tree $T'$ resolves tree $T$ if and only if for all $1 \le j \le n-1$, leaf $l_j$ is placed above $\mroot(T)_y$ if $p(l_j)$ is binary in $T^j$ and $y$ is the sibling of $l_j$ in $T^j$, or it is placed above a node in $\mset(T)_x$ if $p(l_j)$ is non-binary in $T^j$ and $x = \sigma(p(l_i))$. Note that, for convenience, we overload the notation so that $\mroot(T)_y = y$ if $y$ is a leaf.
\end{lemma}
\begin{proof}
    We are going to prove this by induction on $j$. Base case: $j = 1$. In this case, $p(l_1)$ is binary, there is only one option of placement (above $0$), and $\mroot(T)_0 = 0$.

    Consider now some $j > 1$. If $p(l_j)$ is binary in $T^j$ and $y$ is the sibling of $l_j$ in $T^j$, then $\mroot(T)_{y}$ will have the same cluster in $(T')^{j-1}$ as $y$ in $T^{j-1}$ (by the induction hypothesis), and therefore it's the unique placement for $l_j$. Otherwise, if $x = p(l_j)$ is non-binary in $T^j$, consider a placement of $l_j$ above any node $p \in \mset(T)_{x}$. Recall that $\mset(T)_{x}$ is the set of indices in the subtree $B_x$ that resolved the multifurcation (or bifurcation) at $x$ in $(T')^{j-1}$. In this case, $B_x$ grows with new nodes indexed $j$ and $-j$, and it is easy to see that collapsing all internal edges of $B_x$ will restore the original multifurcation; that is, $(T')^j$ resolves $T^j$. Consider now any placement of $l_j$ above a node $p \not \in \mset(T)_{x}$. In $T^j$, consider the path $P$ from $p(l_j)$ to $\mroot(T)_{x}$, and let $u \ne l_j$ be a node such that $u \not \in \mset(T)_{x}$, $u \not \in P$, but $p(u) \in P$. Note that such a node is guaranteed to exist (either $p$ or a child of $p$ always satisfies these criteria). Then, to restore the original multifurcation at $x$, all edges along the path $P$ need to be collapsed; however, in this case, $u$ will become a child of $x$. Since $u$ could not have been an original child of $x$ in $T^j$, $(T')^j$ does not properly resolve $T^j$.
\end{proof}
\begin{corollary}\label{cor:alg2-resolve}
    Trees $\T'$ produced by Algorithm~\ref{alg:resolve} properly resolve the original trees $\T$.
\end{corollary}

From now on, by $OLA(T, \sigma)$ we refer to the $OLA(T', \sigma)$ vector where $T'$ is the resolution of $T$ according to Algorithm~\ref{alg:resolve}. Then, the distances $d_{\sigma}(\T)$ and $\hat{d}_{\sigma}(\T)$ are defined over the resolved vectors $OLA(T, \sigma)$.

\subsection{Algorithm~\ref{alg:resolve} is optimal given an optimal ordering}
The key property of Algorithm~\ref{alg:resolve} is that it maintains our ability to compute the reticulation number and construct MAAFs. Our main result in this section is Theorem~\ref{thm:multi}.

\begin{theorem}\label{thm:multi}
    Let $\T = \set{T_1, \ldots, T_k}$ be a set of (not necessarily binary) trees over the same label set $L$. Let $\sigma^*$ be an optimal leaf-ordering for these trees that minimizes the corrected OLA distance on trees resolved via Algorithm~\ref{alg:resolve}: $\displaystyle\sigma^* = \argmin_{\sigma} \hat{d}_{\sigma}(\T)$. Further, let $\T^* = \set{T_1^*, \ldots, T_k^*}$ be a set of optimal resolutions of trees in $\T$ that minimizes $|\maaf(\T^*)|$. Then,
    \[
        \hat{d}_{\sigma^*}(\T) = |\maaf(\T^*)| - 1 = m(\T).
    \]
    That is, given an optimal ordering, Algorithm~\ref{alg:resolve} will optimally resolve trees in $\T$, minimizing the size of the MAAF and therefore the reticulation number.
\end{theorem}

The lower-bound ($\hat{d}_{\sigma^*}(\T) \ge |\maaf(\T^*)| - 1$) follows directly from Theorem~\ref{thm:OLAk}. That is, let $\sigma$ be a fixed ordering and $\T'$ be the set of resolved trees generated by Algorithm~\ref{alg:resolve}. Then, \[\hat{d}_{\sigma}(\T) = \hat{d}_{\sigma}(\T') \ge |\maaf(\T')| - 1 \ge |\maaf(\T^*)| - 1,\] with the first inequality holding via Theorem~\ref{thm:OLAk}.

To prove the upper-bound, we first modify the set of optimal trees $\T^*$. Let $\F^* = \set{C_i}$ be a MAAF for optimally resolved trees $\T^*$ and $(C_1, \ldots, C_k)$ be a topological ordering of $G(\T^*, \F^*)$. Then, we construct a leaf-ordering $\sigma$ same way as in Theorem~\ref{thm:OLAk}, where we place leaves of $C_i$ before leaves of $C_j$ for $i < j$.

We call a leaf $l_i$ \emph{binary} for tree $T$ if $p(l_i)$ is binary in $T^i$. Otherwise, $l_i$ is \emph{non-binary}. That is, $l_i$ is binary if its addition to $T$ does not create a multifurcation. We call $l_i$ \emph{ambiguous} if it's either non-binary for all $T \in \T$ or it is the first leaf for some $C_j \in \F^*$. Note that these definitions are subject to the ordering $\sigma$. Let $M^-$ denote the set of (negated) mismatched indices of OLA vectors for $\T^*$ (i.e., these are the first indices for each $L(C_j)$, except $0$). We iteratively modify trees $\T^*$ as follows: for each $i \in \set{1, \ldots, n-1}$ such that $l_i$ is \emph{ambiguous} (in order),

\begin{itemize}
    \item[(i)] Let $x=p(l_i)$ in $T^i$, and $\mset(T)_x, \mroot(T)_x$ be defined as before for the partial tree $(T^*)^{i-1}$. If $x$ is non-binary, we perform steps (ii) or (iii).
    
    \item[(ii)] If $l_i \in C_j$ and $l_i$ is the first leaf among $L(C_j)$, then place $l_i$ above $\mroot(T)_x$ for each $T \in \T$. Then, $-i$ becomes the new $\mroot(T)_x$ and we denote the previous value by $p$. For all $j > i$, if the placement for $l_j$ was $p$, then we modify it to $-i$.
    
    \item[(iii)] Otherwise, if $l_i$ is not the first in $L(C_j)$, consider the intersection set $I = \bigcap_{T \in \T}\mset(T)_{x} \setminus M^-$ (by Theorem~\ref{thm:OLAk}, this set must be non-empty). Let $p \in I$ be such that $|p|$ is maximum. Then, we place $l_i$ above the node indexed $p$ in each $T^*$. As above, if $p$ was $\mroot(T)_x$, then for all $j > i$ where $l_j$ was placed above $p$, we modify that placement to $-i$.

    \item[(iv)] Let $\T''$ denote the resulting modified set of trees. We set $\T^* = \T''$ for the next iteration.
\end{itemize}

Next, we show a series of results regarding trees $\T^*$ and $\T''$ at the beginning and end of each iteration, respectively. Note that we start with optimal trees $\T^*$ that resolve $\T$. Then, we show that after each iteration, (a) trees $\T''$ still properly resolve trees $\T$, (b) a MAAF of $\T''$ is of the same size as a MAAF of $\T^*$, and (c) the ordering $\sigma$ follows the topological ordering of a MAAF of $\T''$.

To see that trees $\T''$ resolve $\T$, first note the following:

\begin{observation}\label{obs:resolve}
    Let $T', T''$ be two trees resolving a non-binary tree $T$. For a multifurcated node $v \in V(T)$, let $B'$ and $B''$ be the subtrees of $T'$ and $T''$, respectively, that resolve the multifurcation at $v$. Then, for a fixed leaf-ordering $\sigma$, the internal nodes of $B'$ and $B''$ have the same set of indices; i.e., $\sigma(V_I(B')) = \sigma(V_I(B''))$. 
\end{observation}
\begin{proof}
    Let $u_1, \ldots, u_p$ be the children of $v$. We define $L_i = L(T_{u_i})$. Consider the smallest indices in each $L_i$: let $s_i = \min(\sigma(L_i))$. Then, consider the set $S = \set{-s_i \mid \forall 1 \le i \le p \text{ if } \exists s_j < s_i}$. From the definition of OLA vectors, it is not difficult to see that $\sigma(V_I(B')) = \sigma(V_I(B'')) = S$.
\end{proof}

Next, we show that a tree $T''$ resolves $T$ after a modification at an ambiguous leaf $l_i$.

\begin{lemma}\label{lem:iter-resolve}
    For all $0 \le j \le n-1$, tree $(T'')^j$ resolves $T^j$. Further, for subtrees $(B'')^j$ and $(B^*)^j$ that resolve the multifurcation at $x$ in $(T'')^j$ and $(T^*)^j$, respectively (note that $B''$ and $B^*$ can be empty if $x$ is not yet present in $T^j$), the following holds:
    \begin{itemize}
        \item[(i)] $\sigma(L((B'')^j) = \sigma(L((B^*)^j))$;
        \item[(ii)] For any index $-j \le z \le j$, such that $z \not \in \sigma(V_I((B'')^j)) = \sigma(V_I((B^*)^j))$, the cluster of $z$ in $(T'')^j$ is equivalent to the cluster of $z$ in $(T^*)^j$;
        \item[(iii)] Either $\sigma(r((B'')^j)) = -i\ \&\ \sigma(r((B^*)^j)) = p$ or $\sigma(r(B'')) = \sigma(r((B^*)^j))$ (recall that $p$ is the placement of $l_i$ in $T''$).
    \end{itemize}
\end{lemma}
\begin{proof}
    Recall that $x$ is the parent of $l_i$ in $T^i$. If $x$ is binary, we do not modify the tree $T^*$. Therefore, $T''$ resolves $T$ and all other statements hold.

    Assume now that $x$ is non-binary; then, we potentially modify the tree $T^*$ by placing $l_i$ above a different node $p$ in $\mset(T)_x$. Additionally, if $p$ was $\mroot(T)_x$, we change any subsequent placement of a leaf above $p$ to the placement above $-i$. We prove the statement of the lemma by induction on $j$.

    \textit{Base case:} $j \le i$. For $j < i$ this clearly holds as partial trees $(T^*)^j$ and $(T'')^j$ are identical. For $j=i$, the leaf $l_i$ might have different placements in $T^*$ and $T''$; however, by construction, we make sure that either placement is within the subtree $(B^*)^{i-1} = (B'')^{i-1}$, and therefore $\set{i, -i}$ are added to both $(B^*)^{i-1}$ and $(B'')^{i-1}$. If $l_i$ was added above the root of $(B'')^{i-1}$, then the new root is indexed $-i$. Finally, $i$ is added to the leaf-set of both $(B^*)^{i}$ and $(B'')^{i}$. Note that (ii) holds since the only clusters that potentially get modified are the clusters within $V_I((B'')^j)$.

    \textit{Induction step:} Consider any $j > i$. If $l_j$ was added above a node $z$ outside of $(B^*)^{j-1}$, then by the induction hypothesis, $z$ is outside of $(B'')^{j-1}$ as well (via Observation~\ref{obs:resolve} and property (i)). Since in tree $T''$ the placement of $l_j$ remains above $z$, it is not difficult to see that all three properties (i)-(iii) still hold in this scenario.
    
    Now, assume $l_j$ is placed above some node $y$ in $(B^*)^{j-1}$. We distinguish two cases: when $p(l_j) = x$ in $T^j$ (i.e., $l_j$ must be placed within $(B'')^{j-1}$) and otherwise.
    If $p(l_j) \ne x$, then $y$ could only be a leaf of $(B^*)^{j-1}$ or the root of $(B^*)^{j-1}$. If $y$ is a leaf of $(B^*)^{j-1}$, then it must also be a leaf of $(B'')^{j-1}$, and in both cases, that leaf gets replaced by a node indexed $-j$ within $(B^*)^{j}$ and $(B'')^{j}$ (and properties (i)-(iii) hold). Assume now that $y$ is the root of $(B^*)^{j-1}$, then either $r((B^*)^{j-1}) = r((B'')^{j-1})$ (and then the statement holds) or $r((B'')^{j-1}) = -i$. In the latter case, $r((B^*)^{j-1})$ must be $p$ , the node, above which $l_i$ was placed. As we replaced all placements above $p$ with placements above $-i$, $l_j$ is placed above the root of $(B'')^{j-1}$ (and the statement holds).

    Finally, assume that $p(l_j) = x$. In that case, there must be a node in $(B'')^{j-1}$ indexed $y$ or a node indexed $-i$ if $y = p$. Then, addition of $l_j$ will contribute $\set{-j, j}$ to both the indices of $(B'')^{j}$ and $(B^*)^{j}$. Similarly to above, if $y$ was the root of $(B^*)^{j-1}$, we make sure to add $l_j$ above the root of $(B'')^{j-1}$, and in both cases the new root becomes $-j$ (satisfying (iii)).
\end{proof}

\begin{corollary}\label{cor:iter-resolve}
    After a modification at an ambiguous leaf $l_i$, the resulting trees $\T''$ properly resolve $\T$.
\end{corollary}

Next, we show that there is a MAAF of $\T''$ that is of the same size and structure as a MAAF of $\T^*$.

\begin{lemma}\label{lem:multi-order}
    Let $l_s$ be the first and $l_m$ be a non-first leaf in some $C \in \F^*$. Then, $|OLA(T'', \sigma)_m| \ge s$ and $OLA(T'', \sigma)_m \ne -s$. That is, $OLA(T'', \sigma)_m \in span(\set{l_s, \ldots, l_{m-1}})$.
\end{lemma}
\begin{proof}
    Note that for the starting tree $T^*$, this statement follows directly from Lemma~\ref{lem:C_i}. Therefore, we need to make sure that modifications to the placement of $l_i$ and (potentially) other leaves placed above $p$ maintain this property.

    First, assume that the modified leaf $l_i$ is the first leaf in some $L(C_j)$. Then, moving the placement of $l_i$ corresponds to moving the placement of the entire subtree $C_j$; that is, for each $l_m \in L(C_j), m > i$, its placement stays the same in $T''$ and therefore the lemma holds for such $l_m$. In addition to $l_i$, we also move the placement of any leaf $l_x$ that was placed above $p$ (the sibling of $l_i$ in $(T'')^i$) if $p$ was $\mroot(T)_x$ in $(T^*)^{i-1}$. Consider any such $l_x \in C$ (note that $C$ must follow $C_j$ in topological ordering). It is sufficient to show that  $l_x$ must be the first leaf in $C$ (as all other leaves are placed relative to the first leaf in $C$ and those placements do not change). To show that, assume that node $p$ was created as a result of an addition of some leaf $l_y \in L(C_z)$ ($y < i$). Since $p$ was created before the first leaf of $C$ was added, $p$ cannot be a part of $T''(L(C))$ by Lemma~\ref{lem:C_i}. Recall that by the same lemma, any non-first leaf $l \in C$ must be placed above a node in $T''(L(C))$. Therefore, $l_x$ can only be placed above $p$ if $l_x$ is a first leaf.

    Next, assume that $l_i$ is not a first leaf in $C_j$, but $l_s$ is. The original placement node $p$ must be within $T^*(L(C_j))$. Therefore, by Lemma~\ref{lem:C_i}, $s \le |p| < i$. The placement of $l_i$ might be moved to another node $y$ with $|y| > |p| \ge s$, so the statement holds. If $p$ was the root of a subtree $(B^*)^{i-1}$, then all further placement above $p$ change to $-i$. For all such altered leaves $l_x \in L(C_j)$, the statement clearly holds since $|-i| > s$. Finally, for any leaf $l_x \not\in L(C_j)$, it can only be placed above $p$ if $l_x$ is a first leaf in some $C$.
\end{proof}

\begin{corollary}\label{cor:multi-maaf}
    The forest $\F'' = \set{T''|_{L(C_1)}, \ldots, T''|_{L(C_k)}}$ for any $T'' \in \T''$ is a MAAF for $\T''$. Further, $\sigma$ still follows the topological ordering of the new MAAF.
\end{corollary}

The above corollary follows from the construction of MAAFs from OLA vectors shown in Theorem~\ref{thm:OLAk}.

\begin{observation}\label{obs:T''}
    The trees $\T''$ at the end of all iterations are equivalent to the trees $\T'$ generated by Algorithm~\ref{alg:resolve} given the ordering $\sigma$.
\end{observation}
\begin{proof}
    It is not difficult to see that our iterative procedure that tweaks trees $\T^*$ makes sure that, at every leaf index, the placement choices from Algorithm~\ref{alg:resolve} are identical to the placements within trees $\T''$.
\end{proof}

Observation~\ref{obs:T''} and Corollary~\ref{cor:multi-maaf} imply that Algorithm~\ref{alg:resolve} optimally resolves trees in $\T$ given an optimal ordering $\sigma^*$.

\subsection{Complexity analysis} Algorithm~\ref{alg:resolve} adds a layer of complexity over the original OLA algorithm due to the need to keep track of multifurcations and how they are getting resolved.
\begin{lemma}\label{lem:complexity}
    Let $k = |\T|$, $n$ be the number of leaves, and $m$ be the size of the largest multifurcation across trees in $\T$. Then, Algorithm~\ref{alg:resolve} runs in $O(kn\cdot m\log{m})$ time. 
\end{lemma}
\begin{proof}
    The added complexity comes from $\mset(T)_i$ manipulations ($O(\log m)$ per each insertion/ deletion/check if sets are implemented via self-balancing trees). Additionally, Line~\ref{ln:inter} requires at most $k-1$ set intersections where each set is of size at most $m$, which implies $O(km\log m)$ complexity.
\end{proof}

\section{Conclusion}
In this work, we show how OLA vectors can be used to estimate the reticulation number and construct phylogenetic networks on large datasets. In Theorem~\ref{thm:OLAk}, we show that for a set of binary trees $\T$, the reticulation number is equivalent to the corrected OLA distance under an optimum leaf-ordering that minimizes that distance. In particular, we show that there is a direct connection between OLA vectors and acyclic agreement forests. Given an AAF $\F$, we can generate a leaf-ordering that induces an OLA distance of $|\F| - 1$, and given a leaf-ordering with the corresponding corrected OLA distance $d$, we can construct an AAF of size $d + 1$ as shown in Fig.~\ref{fig:AAF-from-OLA}. As in practice $\T$ is often non-binary, especially if we want to minimize tree estimation error, we propose Algorithm~\ref{alg:resolve} that resolves the trees while optimizing the corrected OLA distance. In Theorem~\ref{thm:multi}, we show that this resolution algorithm also minimizes the size of a MAAF of the resolved trees.

These results imply that solving the MAAF problem is equivalent to finding an optimal leaf ordering that minimizes the corrected OLA distance. This connection is especially practical when working with fast-evolving microorganisms, such as RNA viruses. In these cases, the sample collection dates provide a natural ordering of leaves, which can be expected to be close to optimal, especially when the sampling density is high such as for SARS-CoV-2 or influenza A viruses. Further empirical evaluation of our proposed method for reticulation number estimation and phylogenetic network reconstruction is required and will be carried out in future studies.

Whereas we show that minimizing the corrected OLA distance is an NP-hard problem as it is equivalent to the MAAF problem, the question of whether minimizing the Hamming OLA distance over all leaf permutations is NP-hard remains open. It may be possible to define a relaxed version of an acyclic agreement forest to match the Hamming OLA distance -- in particular, in such a forest, full paths of agreement subtrees can be ``moved'' between trees as long as the agreement subtrees maintain the same order (e.g., see Fig.~\ref{fig:correction}B-C).

\section*{Code availability}
All algorithms described in this work for reticulation number estimation, multifurcation resolution, and AAF reconstruction are available on GitHub at \url{https://github.com/flu-crew/OLA-Net}.

\section*{Acknowledgments}
This work was supported in part by the USDA-ARS (ARS project number 5030-32000-231-000D); USDA-APHIS (ARS project number 5030-32000-231-104-I, 5030-32000-231-111-I); the National Institute of Allergy and Infectious Diseases, National Institutes of Health, Department of Health and Human Services (Contract No. 75N93021C00015); the Centers for Disease Control and Prevention (contract numbers 21FED2100395IPD, 24FED2400250IPC); and the SCINet project of the USDA-ARS (ARS project number 0500-00093-001-00-D). The funders had no role in study design, data collection and interpretation, or the decision to submit the work for publication. Mention of trade names or commercial products in this article is solely for the purpose of providing specific information and does not imply recommendation or endorsement by the USDA or CDC. USDA is an equal opportunity provider and employer.  

\bibliography{references}
\bibliographystyle{abbrv}

\end{document}